\documentclass[manuscript,screen]{acmart}
\usepackage{diagbox}
\usepackage[ruled,vlined]{algorithm2e}

\newcommand{\argmin}[1]{\underset{#1}{\operatorname{arg}\,\operatorname{min}}\;}
\newcommand{\R}{\mathbb{R}}
\newcommand{\opt}{\mathrm{opt}}

\newcommand{\vc}[1]{\mathrm{vc}\left( #1 \right)}

\newcommand{\CR}{\mathrm{CR}}
\newcommand{\PI}{\mathrm{PI}}

\usepackage{bm}
\AtBeginDocument{%
  \providecommand\BibTeX{{%
    \normalfont B\kern-0.5em{\scshape i\kern-0.25em b}\kern-0.8em\TeX}}}

\setcopyright{acmcopyright}
\copyrightyear{2022}
\acmYear{2022}
\acmDOI{XXXXXXX.XXXXXXX}

\begin{document}

\title{Prediction Intervals for Simulation Metamodeling}


\author{Henry Lam}
\affiliation{%
  \institution{Columbia University}
  \city{New York}
  \country{USA}
}

\author{Haofeng Zhang}
\affiliation{%
  \institution{Columbia University}
  \city{New York}
  \country{USA}
}

\renewcommand{\shortauthors}{H. Lam and H. Zhang}

\begin{abstract}
Simulation metamodeling refers to the construction of lower-fidelity models to represent input-output relations using few simulation runs. Stochastic kriging, which is based on Gaussian process, is a versatile and common technique for such a task. However, this approach relies on specific model assumptions and could encounter scalability challenges. In this paper, we study an alternative metamodeling approach using prediction intervals to capture the uncertainty of simulation outputs. We cast the metamodeling task as an empirical constrained optimization framework to train prediction intervals that attain accurate prediction coverage and narrow width. We specifically use neural network to represent these intervals and discuss procedures to approximately solve this optimization problem. We also present an adaptation of conformal prediction tools as another approach to construct prediction intervals for metamodeling. Lastly, we present a validation machinery and show how our method can enjoy a distribution-free finite-sample guarantee on the prediction performance. We demonstrate and compare our proposed approaches with existing methods including stochastic kriging through numerical examples.
\end{abstract}

\ccsdesc[500]{Computing methodologies~Simulation types and techniques}



\keywords{prediction interval, simulation metamodeling, conformal prediction, uncertainty quantification, neural network, deep learning}

\maketitle

\section{INTRODUCTION}
\label{sec:intro}

Stochastic simulation aims to compute output summaries from complex stochastic models that could not be handled analytically. Often the random output quantity, say $Y(x)$, depends on a certain input factor or design parameter $x$, and $y(x)=E[Y(x)]$ is called the (mean) response surface. 
Building a response surface serves a range of usages across sensitivity analysis and optimization, or simply for visualization or an understanding of the relation landscape. It is especially useful when the simulation run is computationally expensive. In this situation, we may not be able to run enough simulations to estimate the value of $y(x)$ at each prediction point $x$ whenever the need comes up. It is thus useful to get an approximate response surface by  running a number of simulation runs at possibly various points of $x$ in advance, and building this surface using regression tools. This task is often known as simulation metamodeling in the literature \cite{barton2006metamodel,staum2009better}.

A benchmark approach in simulation metamodeling is stochastic kriging (SK) \cite{ankenman2008stochastic,ankenman2010stochastic}, which is a versatile technique applicable to general nonlinear input-output relations based on Gaussian process (GP) regression. SK can be viewed as a generalization of kriging \cite{stein1999interpolation,kleijnen2009kriging}, which captures only epistemic or extrinsic uncertainty (i.e., error coming from model assumption or fitting error), to include aleatory or intrinsic uncertainty (i.e., the stochasticity of the system itself that cannot be washed away with enough fitting data), by including extra variances in the Gaussian process. In the literature, the estimator produced by SK typically focuses on mean response surface estimation (i.e., the $y(x)$ above) or quantile-based response measures \cite{chen2013building,bekki2014steady,chen2016efficient}. Enhanced estimators for the measurement of aleatory uncertainty, such as variance or prediction intervals, could also be produced as a by-product. In this paper, much like SK, we will focus on techniques that capture both aleatory and epistemic uncertainties in metamodeling.

Despite the strengths and popularity of SK, it incurs some potential limitations. First is that it relies on the normality of aleatory uncertainties. For mean response estimation, this assumption can potentially be relaxed by running a large number of simulation runs per design point and invoking the central limit theorem. However, this approach cannot capture the distributional shape in the aleatory uncertainty of $Y(x)$, and is vulnerable to misspecification of the GP \citep{staum2009better}. Moreover, to estimate the aleatory variance to be incorporated into the GP, one typically plugs in the sample variance at each design point, which 
enforces the use of many runs per point that could add computational demand. Second, SK prediction involves matrix inversion that could lead to two computational issues: a) calculating matrix inversion is computationally demanding when the matrix size is large; b) near-singularity of the matrix could arise if simulation outputs of some of design points are too highly
correlated, e.g., when design points are too close to each other \cite{staum2009better}. In general, when a large number of design points are available, we should expect a method to perform well because of sufficient data. Unfortunately, in this ideal case both computational issues of SK arise and make it less applicable. Third, while SK enjoys attractive Bayesian interpretation, little is known about the finite-sample frequentist guarantee on its prediction performance, which could depend heavily on the prior correlation structure imposed on the GP.

Motivated by the above limitations, in this paper we propose an alternative simulation metamodeling method built on prediction intervals. Instead of targeting a specific summary statistic (e.g., expectation) of the conditional distribution of $y|x$, This method outputs, for each design point $x$, a lower bound $L(x)$ and upper bound $U(x)$ that cover the random output $Y(x)$ with high probability in some sense. To be more precise, we utilize the notion of so-called \textit{expected coverage rate}, a recent criterion for constructing high-quality prediction intervals \cite{khosravi2010lower,pearce2018high,rosenfeld2018discriminative,chen2021learning}. The expected coverage rate measures the probability of $[L(X),U(X)]$ covering $Y(X)$, when $X$ and $Y$ are both viewed as random. Here, the randomness of $Y$ is simply the aleatory uncertainty or stochasticity of simulation, and the randomness of $X$ is chosen by the user, for instance uniform over the input domain (but can be a more general distribution). To fit globally and equally for each $X$, one common objective is to minimize the \textit{integrated mean squared error} (IMSE) \cite{sacks1989design,ankenman2010stochastic}, which corresponds to assignment of $X$ using a uniform distribution in our approach. Nonetheless, we also allow the possibility of a more general distribution on $X$. This will 
enable users to assign different weights on different input points and similarly correspond to a ``weighted'' IMSE as a criterion.




Let us explain the high-quality criterion mentioned above.
The expected coverage rate itself cannot well characterize the performance of prediction intervals since a sufficiently wide prediction interval can cover any random output $Y(x)$. Thus interval width has been a common metric to measure the conservativeness of prediction intervals \cite{barber2019predictive,zhang2019random}. From this view, constructing a high-quality prediction interval can be formalized as a constrained stochastic optimization problem that optimizes the expected interval width while maintaining the expected coverage rate \cite{rosenfeld2018discriminative,chen2021learning}. 

In this work, we study a neural-network-based method to train a high-quality prediction interval based on this constrained optimization problem. While this problem is difficult to solve directly, it provides a framework for empirical training as follows. First, we consider an empirical formulation to approximate this problem with training data and  derive an empirical loss function from the Lagrangian, with the Lagrangian multiplier treated as a hyper-parameter. With the lower and upper bounds represented by a neural network, we optimize these bounds via gradient descent. Finally, the Lagrangian multiplier needs to be properly calibrated, and we propose an easy-to-implement validation strategy to select the best hyperparameter value that can guarantees the coverage attainment with prefixed confidence. This approach generalizes the previous work \cite{chen2021learning}, which focuses on supervised learning problems, to simulation metamodeling, where we allow multiple simulation runs per design point as commonly used in practice. In addition, as a comparison, we also study another prediction interval construction approach, by adapting the line of work on conformal prediction \cite{vovk2005algorithmic} into the setting of simulation metamodeling and discuss its pros and cons.


Our proposal has the following advantages: (1) It is arguably computationally scalable in the sense that its running time is not sensitive to the sample size.
The high efficiency of our method stems from recent advances in training a neural network such as Adam gradient descent \cite{kingma2014adam} and mini-batch training manner \cite{li2014efficient}.
(2) Our method is distribution-free in the sense that we do not impose distribution assumption and structure on the data. This avoids the common mismatches between data and models, for example, if the GP structure assumption is violated, or if the estimations of the mean or variance of GPs are unreliable. (3) Our method does not require multiple simulation runs per input point, which is essential for SK to make reliable estimation of aleatory variance. (4) Our method enjoys finite-sample performance guarantees on the overall coverage of the simulation output and the interval width. On the other hand, a limitation of our approach is the lack of specific guarantee for any particular design point, as our construction focuses on coverage and width quality on average throughout the design space. When using metamodels for optimization purpose, this issue could be especially relevant, and future work would study and test the performances and potential remedies of this limitation.

We close this introduction by reviewing several mainstream approaches for constructing prediction intervals. The first is to estimate quantiles and convert them into intervals. This approach includes classical quantile regression \cite{koenker2001quantile}, quantile regression forests \cite{meinshausen2006quantile} and quantile-based stochastic kriging \cite{chen2013building,bekki2014steady}. However, little is known about their finite-sample performance.  
The second is conformal prediction, which utilizes quantile-based derivations  to generate distribution-free prediction intervals, some of which enjoy finite-sample coverage guarantees \cite{shafer2008tutorial,lei2018distribution}. These guarantees, however, are on the joint distribution of both the training and testing data that are weaker than ours, as we will discuss in further detail. The third is using deep learning. Neural networks have achieved impressive performance in constructing high-quality prediction intervals \cite{khosravi2010lower,pearce2018high,chen2021learning}, and training a neural network now has become efficient and computationally cheap. 
Moreover, a neural network provides a more general class of functions rather than linear functions or quadratic functions that are widely used as the ``trend'' functions in metamodel \cite{law2000simulation} and thus it can handle more sophisticated data structure. However, most of these works are empirical. Our work follows this stream but enjoys a stronger statistical guarantee about coverage than conformal prediction and empirically produces less conservative intervals. Lastly, while our work focuses primarily on constructing prediction intervals to quantify uncertainty, in situations where a mean response prediction is also desired, we can integrate our approach with other works such as \cite{thiagarajan2020building} to estimate means and prediction intervals simultaneously.

\section{STOCHASTIC KRIGING}\label{sec:problem}

Before introducing our method, we first briefly review the predominant approach in simulation metamodeling, SK. Although the literature has primarily focused on mean response prediction, SK can also be used naturally to generate prediction intervals by leveraging the posterior distribution. 

We consider the following setting. We have a design variable $X\in\mathcal X\subset\mathbb R^d$ and a simulation output $Y\in\mathcal Y\subset\mathbb R$. Suppose for each input value $x_i\in \mathcal{X}$ in $\{x_1,x_2,\ldots,x_n\}$, we run simulation replications of size $r_i$ at $x_i$ and obtain the simulation output $y_{i,j}\in \mathbb R$ where $j=1,..,r_i$. Let the sample mean of responses at $x_i$ be
$\bar{y}_i=\frac{1}{r_i}\sum_{j=1}^{r_i} y_{i,j}$
and $\bar{\mathbf y}=(\bar{y}_1,\bar{y}_2,\ldots,\bar{y}_n)^T$ in short.

SK makes the following assumption on the simulation data: Specifically, a stochastic simulation's output is the sum of the extrinsic Gaussian process $M$, whose realization is the response surface, and an independent intrinsic noise $\varepsilon$ \cite{ankenman2010stochastic}:
\begin{equation} \label{SK}
Y_j(x) = M(x) +\varepsilon_j(x),\ M \sim GP(\mu,\sigma^2),\  \varepsilon_j \sim GP(0, c)    
\end{equation}
where $GP(\mu,\sigma^2)$ is a Gaussian process (GP) with mean $\mu(x)$ and covariance between any two points $x$ and $x'$ given by $\sigma^2(x,x')$. The two GPs $M$ and $\varepsilon_j$ are assumed to be independent. The intrinsic noise
$\varepsilon_1(x), \varepsilon_2(x), \ldots$ at the same input point $x$ is naturally independent and identically distributed (i.i.d.) across
replications. Across different $x$'s, it follows the $GP(0,c)$ structure. In this model, $\mu$ represents the input-output trend, $\sigma^2$ represents the epistemic uncertainty and $c$ represents the aleatoric uncertainty. In particular, if $\varepsilon_j \equiv 0$, then this reduces to kriging \cite{stein1999interpolation}. If $c(x,x)\equiv c_0$ and $c(x,x')= 0$ for $x\ne x'$ (i.e., i.i.d. Guassian noise at every input point), then this is kriging with measurement error \cite{cressie1993statistics}. The normality of the intrinsic noise $\epsilon_j(x)$ is an additional assumption used by the standard SK model to justify that the SK prediction is the posterior mean \cite{staum2009better,ankenman2010stochastic,chen2014stochastic}. Consequently, standard SK does not directly estimate the distributional shape of aleatory uncertainty.  

To implement the above estimation in practice, we need to assume some structure on the mean and covariance functions. The original work of SK \cite{ankenman2010stochastic} suggests that one could use, for example, $\mu(x)=\bm{f}(x)^T\beta$ where $\bm{f}$ is a vector of known functions of $x$ and $\beta$ is a vector of unknown parameters of compatible dimension, and $\sigma^2(x, x') = \tau^2 r_\theta(|x-x'|)$ where $r_\theta$ is chosen from a set of functions parameterized by some unknown parameters $\theta$. These parameters can be calibrated via maximum likelihood estimation.

The SK prediction at any point $\tilde{x}$ is defined as follows \cite{ankenman2010stochastic}:
$$\tilde{\mu}(\tilde{x})=\mu(\tilde{x})+\sigma^2(\tilde{x},{\mathbf x})(\sigma^2({\mathbf x},{\mathbf x})+c({\mathbf x},{\mathbf x}))^{-1}(\bar{\mathbf y}-\mu({\mathbf x})).$$
Here we use the shorthand
$$\sigma^2(\tilde{x},{\mathbf x})=(\sigma^2(\tilde{x},x_1), \ldots, \sigma^2(\tilde{x},x_n)),\ \ \sigma^2({\mathbf x},\tilde{x})=(\sigma^2(x_1,\tilde{x}), \ldots, \sigma^2(x_n,\tilde{x}))^T=\sigma^2(\tilde{x},{\mathbf x})^T,$$
$$\sigma^2({\mathbf x},{\mathbf x})=(\sigma^2(x_i,x_j))_{i=1,\ldots,n;j=1,\ldots,n}, \ \ \mu({\mathbf x})=(\mu(x_1),\mu(x_2),\ldots,\mu(x_n))^T$$
$$c({\mathbf x},{\mathbf x})=\Big(c(x_i,x_j)\big(\textbf{1}_{(i\ne j)}+\frac{1}{r_i}\textbf{1}_{(i= j)}\big)\Big)_{i=1,\ldots,n;j=1,\ldots,n},$$
where ${\mathbf x}=(x_1,x_2,\ldots,x_n)^T$ and $\bar{\mathbf y}=(\bar{y}_1,\bar{y}_2,\ldots,\bar{y}_n)^T$. 

If we restrict ourselves to only use linear predictors of the form $\lambda_0(\tilde{x})+\bm{\lambda}(\tilde{x})^T\bar{\textbf{y}}$
to predict the true mean response, then Appendix EC.1 of \cite{ankenman2010stochastic} shows that the SK prediction is the MSE-optimal predictor. In addition, in general (i.e., if we relax linearity) the SK prediction can be alternatively interpreted as the posterior mean with GP assumptions on the aleatory and
epistemic uncertainty in (\ref{SK2}). To see this, we view the model (\ref{SK}) as the prior distribution
over possible functions $Y(x)$ and then use the observed simulation data $(x_i,y_{i,j})$ to update this distribution. Provided the validity of the specified form of uncertainty, SK naturally leads to generation of a prediction interval (in a Bayesian sense) under the posterior distribution. 

\begin{lemma} \label{SK2}
Consider the SK model (\ref{SK}) as a prior belief. The simulation data we observe is $(x_i,y_{i,j})$ where $j=1,\ldots,r_i$ and $i=1,\ldots,n$. Then the posterior distribution of simulation output $\tilde{y}$ at any point $\tilde{x}$ is a Gaussian distribution with mean
$$\tilde{\mu}(\tilde{x})=\mu(\tilde{x})+\sigma^2(\tilde{x},{\mathbf x})(\sigma^2({\mathbf x},{\mathbf x})+c({\mathbf x},{\mathbf x}))^{-1}(\bar{\mathbf y}-\mu({\mathbf x}))$$
and variance
$$\tilde{\sigma}^2(\tilde{x})=\sigma^2(\tilde{x},\tilde{x})+\sigma^2(\tilde{x},{\mathbf x})(\sigma^2({\mathbf x},{\mathbf x})+c({\mathbf x},{\mathbf x}))^{-1}\sigma^2({\mathbf x},\tilde{x}).$$
\end{lemma} 

The proof of Lemma \ref{SK2} is given in Appendix \ref{sec:proofs}. It is then easy to see that a Type-IV prediction interval with prediction level $1-\alpha$ (see Definition \ref{def:PI} below) based on the SK model (\ref{SK}) is 
$$[\tilde{\mu}(\tilde{x})-\lambda \tilde{\sigma}(\tilde{x}),\tilde{\mu}(\tilde{x})+\lambda \tilde{\sigma}(\tilde{x})]$$
where $\lambda$ is the $1-\alpha/2$ quantile of the standard normal distribution.

The SK prediction requires us to evaluate the matrix inversion $(\sigma^2({\mathbf x},{\mathbf x})+c({\mathbf x},{\mathbf x}))^{-1}$ as in Lemma \ref{SK2}. To slightly reduce the computational complexity, \cite{staum2009better} suggests to apply cholesky factorization on the matrix $\sigma^2({\mathbf x},{\mathbf x})+c({\mathbf x},{\mathbf x})$ since it is non-negative, which takes complexity $O(n^3)$ in the number of design points $n$. This step is computationally demanding especially when the matrix size is large. Another computational issue is that numerical problems about near-singularity
of $\sigma^2({\mathbf x},{\mathbf x})+c({\mathbf x},{\mathbf x})$ arise if simulation outputs of some of design points are too highly correlated \cite{staum2009better}.
 
Lemma \ref{SK2} assumes that the structures of $\mu$, $\sigma^2$ and $c$ are known to get the prediction. In practice, we should also estimate $\mu$, $\sigma^2$ and $c$ based on the data. \cite{ankenman2010stochastic} provide some guidance. For example, they suggest to plug in the sample variance at each design point to estimate the aleatoric variance $c({\mathbf x},{\mathbf x})$. Since one or very few responses at each design point will make the variance estimation unreliable, this requires us to have access to sufficiently many responses at the design point, which could add simulation demand. We also point out that the posterior distribution from Lemma \ref{SK2} is based on the normality assumption of the intrinsic noise in (\ref{SK}). It is beyond the scope of standard SK if the normality assumption is violated by the actual data, which makes SK vulnerable to misspecification of the GP \citep{staum2009better}.

In general, SK performance depends on the appropriate choice of these structures in the model and little is known about its finite-sample guarantee on the prediction performance. In the following, we consider a prediction interval approach to capture epistemic and aleatory uncertainties in metamodeling which enjoys performance guarantees.

\section{PREDICTION INTERVALS} \label{sec:PI}

We introduce another way to communicate predictive uncertainty in simulation metamodeling, which is called a prediction interval. A prediction interval gives a bracket of simulation outputs rather than only their mean response. To be rigorous, we assume that the simulation data $(x,y)$ follows a general joint distribution $\pi$. Note that unlike Equation (\ref{SK2}) in Section \ref{sec:problem}, we do not make any assumptions about the distribution $\pi$ throughout this section.

 
\begin{definition} \label{def:PI}
An interval $[L(x),U(x)]$, where both $L,U:\mathcal{X}\to \R$, is called a Type-IV prediction interval with prediction level $1-\alpha$ ($0\le \alpha\le 1$) if 
\begin{equation}
\mathbb{P}_{\pi(Y|X)}(Y\in[L(X),U(X)]|L,U,X)\geq 1-\alpha, \quad \forall \text{ a.e. } X
\end{equation}
where $\mathbb{P}_{\pi(Y|X)}$ denotes the probability with respect to the conditional distribution $\pi(Y|X)$ with $L$, $U$, $X$ fixed.
\end{definition}

The terminology of Type-IV (and Type-II, Type-I below) prediction intervals is suggested in \cite{zhang2019random}. Although this definition of prediction interval is appealing and desirable, it is difficult to measure and achieve in general without assuming some simple structure on the joint distribution $\pi$. Quantile regression
forests \cite{meinshausen2006quantile}, quantile-based stochastic kriging \cite{chen2013building}, and random forest prediction intervals \cite{zhang2019random} could be applied for this target but no finite-sample guarantee is known. Therefore, we focus on the following relaxation of definition of prediction intervals \cite{rosenfeld2018discriminative,chen2021learning}. 

\begin{definition}
An interval $[L(x),U(x)]$, where both $L,U:\mathcal{X}\to \R$, is called a Type-II prediction interval with prediction level $1-\alpha$ ($0\le \alpha\le 1$) if 
\begin{equation} \label{PI1}
\mathbb{P}_{\pi}(Y\in[L(X),U(X)]|L,U)\geq 1-\alpha    
\end{equation}
where $\mathbb{P}_{\pi}$ denotes the probability with respect to the joint distribution $\pi$ on $(X,Y)$ with $L$, $U$ fixed.
\end{definition}

For example, a prediction interval with $95\%$ prediction level is expected to cover at least $95\%$ of the simulation outputs under $\pi$. This definition is in parallel with the IMSE criterion \cite{ankenman2010stochastic} since the expected coverage rate is defined globally for all $X$.  If the prediction interval at $x$ has a very high width $U(x)-L(x)$, we can interpret that it has high uncertainty at point $x$. In addition, there is another type of definition of prediction intervals, which is a further relaxation of Type-II prediction intervals. This definition is widely-used in the area of conformal prediction \cite{lei2018distribution}. 

\begin{definition}
An interval $[L(x),U(x)]$, where both $L,U:\mathcal{X}\to \R$, is called a Type-I prediction interval with prediction level $1-\alpha$ ($0\le \alpha\le 1$) if
\begin{equation} \label{PI2}
\mathbb{P}(Y\in[L(X),U(X)])\geq 1-\alpha,
\end{equation}
where $\mathbb{P}$ denotes the probability with respect to the joint distribution on both future point $(X,Y)$ and $L$, $U$ (training data).
\end{definition}

Unlike Equation \eqref{PI1}, Equation \eqref{PI2} takes into account the randomness in $L$ and $U$ since $L$ and $U$ are constructed by the (random) training data. It is easy to see that a Type-II prediction interval with prediction level $1-\alpha$ must be a Type-I prediction interval with prediction level $1-\alpha$ by taking expectation with respect to $L$, $U$ in Equation \eqref{PI1}. In the following, we introduce two lines of work to generate prediction intervals borrowed from machine learning into simulation metamodeling, and then propose a validation algorithm to optimally select the prediction interval model.

\section{CONFORMAL PREDICTION FOR PREDICTION INTERVALS} \label{sec:CP}

Conformal prediction is a class of generic approaches to construct distribution-free Type-I prediction intervals. 
The original conformal prediction \cite{vovk2005algorithmic} incurs a high computational cost since it requires retraining for each new observed $x$. The split conformal prediction \cite{lei2015conformal} improves the computational efficiency and offers an assumption-free guarantee but comes at the cost of higher variance. More recently, split conformal quantile regression combines the quantile regression and split conformal prediction to output intervals that are adaptive
to heteroscedasticity whereas the standard split conformal prediction cannot \cite{romano2019conformalized}.
Some of recent conformal prediction approaches do not have finite-sample coverage guarantee while they are more efficient or generate narrower intervals. These methods include the Jackknife conformal prediction \cite{lei2018distribution}, K-fold conformal prediction \cite{vovk2015cross}, Jackknife+ conformal prediction and CV+ conformal prediction \cite{barber2019predictive}. 

As concrete algorithms, we review two widely-used conformal prediction approaches: split conformal prediction and split conformal quantile regression, and then adapt them into the setting of simulation metamodeling. To begin with, consider the standard setting in conformal prediction where we have observed i.i.d. training data $\{(x_i.y_i):i=1,\ldots,n\}$ and aim to construct a Type-I prediction interval at any point $\tilde{x}$ with prediction level $1-\alpha$. First, we randomly split $\{1, \ldots , n\}$ into two disjoint sets $\mathcal{I}_1$ and $\mathcal{I}_2$. 

In split conformal prediction, given any regression algorithm $\mathcal{A}$ (such as linear regression using least squares), a regression model is fit to $\{(x_i,y_i):i\in \mathcal{I}_1\}$:
$$\mathcal{A} (\{(x_i,y_i):i\in \mathcal{I}_1\}) \to \hat{\mu}(x)$$
where $\hat{\mu}(x)$ is the estimated regression function.
Define
$R_i = |y_i - \hat{\mu}(x_i)|$, $i \in  \mathcal{I}_2$, the residuals on $\mathcal{I}_2$. Then compute the  
$Q_{1-\alpha}(R, \mathcal{I}_2) := (1 -\alpha)(1 + 1/|\mathcal{I}_2|)$-th empirical quantile of the set $\{R_i: i \in \mathcal{I}_2\}$ (the $\bar{\alpha}$-th empirical quantile of a set of size $s$ is defined as the $\lceil \bar{\alpha}s \rceil$-th smallest value in this set). 
Finally, the prediction interval at any point $\tilde{x}$ is given by
$$[\hat{\mu}(\tilde{x}) -Q_{1-\alpha}(R, \mathcal{I}_2), \hat{\mu}(\tilde{x}) +Q_{1-\alpha}(R, \mathcal{I}_2)].$$
\cite{lei2018distribution} have shown that this interval satisfies Equation \eqref{PI2}.

In split conformalized quantile regression, given any quantile regression algorithm $\mathcal{A}'$ (such as quantile regression forests \cite{meinshausen2006quantile}), a quantile regression model is fit to $\{(x_i,y_i):i\in \mathcal{I}_1\}$:
$$\mathcal{A}' (\{(x_i.y_i):i\in \mathcal{I}_1\}) \to (\hat{q}_{\alpha/2}(x),\hat{q}_{1-\alpha/2}(x))$$
where $\hat{q}_{\alpha}(x)$ is the estimated $\alpha$-quantile function.
Define
$E_i= \max\{\hat{q}_{\alpha/2} (x_i) - y_i
, y_i - \hat{q}_{1-\alpha/2}(x_i)\}$, $i \in  \mathcal{I}_2$, the residuals on $\mathcal{I}_2$. Then compute
$Q'_{1-\alpha}(E, \mathcal{I}_2) := (1 -\alpha)(1 + 1/|\mathcal{I}_2|)$-th empirical quantile of the set $\{E_i: i \in \mathcal{I}_2\}$.
Finally, the prediction interval at any point $\tilde{x}$ is given by
$$[\hat{q}_{\alpha/2} (\tilde{x}) -Q'_{1-\alpha}(E, \mathcal{I}_2), \hat{q}_{1-\alpha/2} (\tilde{x}) +Q'_{1-\alpha}(E, \mathcal{I}_2)].$$
\cite{romano2019conformalized} have shown that this interval satisfies Equation \eqref{PI2}.

Now we turn to the setting of simulation metamodeling in a frequentist perspective: Suppose the simulation data are given by $(x_i, y_{i,j})$ where $i=1,\ldots,n$ and for each $i$, $j=1,\ldots,r_i$ where $r_i$ is the simulation replications at $x_i$. We assume that for any sequence $j_i\in \{1,\ldots,r_i\} (i=1,\ldots,n)$, the data $(x_1, y_{1,j_1}), \ldots, (x_n, y_{n,j_n})$ are i.i.d. drawn from the joint distribution $\pi$ and for a fixed $x_i$, the corresponding simulation outputs $y_{i,1},\ldots,y_{i,r_i}$ are conditional i.i.d. drawn from the conditional distribution of $y|x_i$. The former assumption is standard in statistical learning and the latter assumption is adopted in SK \cite{ankenman2010stochastic} as the intrinsic noise is i.i.d. in \eqref{SK}. In this setting (multiple simulation runs per design point), the standard split conformal prediction described above will fail and thus we propose an adaption in Algorithm \ref{algo:SCP}. The following theorem establishes the statistical effectiveness of Algorithm \ref{algo:SCP}.

\begin{theorem} \label{thm:SCP}
The interval output by Algorithm \ref{algo:SCP} is a Type-I prediction interval with prediction level $1-\alpha$. More specifically, for a new i.i.d. draw $(\tilde{x}, \tilde{y})\sim \pi$, we have that 
\begin{equation} \label{equ:SCP1}
\mathbb{P}(\tilde{y}\in[\hat{\mu}(\tilde{x}) -Q_{1-\alpha}(R, \mathcal{I}_2), \hat{\mu}(\tilde{x}) +Q_{1-\alpha}(R, \mathcal{I}_2)])\geq 1-\alpha,
\end{equation}
where the expectation is taken with respect to both $(\tilde{x}, \tilde{y})$ and data $\mathcal{D}_2$ (or both $(\tilde{x}, \tilde{y})$ and the interval, since the former implies the latter).
\end{theorem}

The property \eqref{equ:SCP1} follows by the symmetry between the residual at $(\tilde{x}, \tilde{y})$ and those at $(x_i, y_{i,1}), i \in \mathcal{I}_2$. The detailed proof of Theorem \ref{thm:SCP} is given in Appendix \ref{sec:proofs}. We can similarly propose an adaption of split conformalized quantile regression in Algorithm \ref{algo:SCQR}. Its statistical effectiveness can be established similarly as in Theorem \ref{thm:SCP}.

\begin{algorithm}
\caption{Split Conformal Prediction for Simulation Metamodeling}
\label{algo:SCP}
\DontPrintSemicolon
\SetKwInOut{Input}{Input}\SetKwInOut{Output}{Output}

\textbf{Input:} training data $\mathcal{D}_{tr}=\{(x_i, y_{i,1}),\ldots, (x_i, y_{i,r_i}):i=1,\ldots,n\}$, and a regression algorithm $\mathcal{A}$.
\BlankLine
\textbf{Procedure:}\;
\textbf{1.}~Split $\{1, \ldots, n\}$ into two disjoint sets $\mathcal{I}_1$ and $\mathcal{I}_2$. Let $\mathcal{D}_2=\{(x_i,y_{i,1}):i\in\mathcal{I}_2\}$ (WLOG we write $y_{i,1}$ but it can be replaced by $y_{i,j_i}$ where $j_i$ is any index in $1,\ldots,r_i$) and $\mathcal{D}_1=\mathcal{D}_{tr}\backslash \mathcal{D}_2$. If $r_i\ge 2$ for any $i$, then we can, for instance, take $\mathcal{I}_1=\emptyset$ and $\mathcal{I}_2=\{1, \ldots, n\}$.\;

\textbf{2.}~Fit a regression model to $\mathcal{D}_1$:
$$\mathcal{A} (\mathcal{D}_1) \to \hat{\mu}(x).$$

\textbf{3.}~Define
$R_i = |y_{i,1} - \hat{\mu}(x_i)|$, $i \in  \mathcal{I}_2$, the residuals on $\mathcal{I}_2$. Then compute the 
$Q_{1-\alpha}(R, \mathcal{I}_2) := (1 -\alpha)(1 + 1/|\mathcal{I}_2|)$-th empirical quantile of the set $\{R_i: i \in \mathcal{I}_2\}$.
\BlankLine

\Output{The prediction interval at point $\tilde{x}$ is given by
$$[\hat{\mu}(\tilde{x}) -Q_{1-\alpha}(R, \mathcal{I}_2), \hat{\mu}(\tilde{x}) +Q_{1-\alpha}(R, \mathcal{I}_2)].$$}

\end{algorithm}

\begin{algorithm}
\caption{Split Conformalized Quantile Regression for Simulation Metamodeling}
\label{algo:SCQR}
\DontPrintSemicolon
\SetKwInOut{Input}{Input}\SetKwInOut{Output}{Output}

\textbf{Input:} training data $\mathcal{D}_{tr}=\{(x_i, y_{i,1}),\ldots, (x_i, y_{i,r_i}):i=1,\ldots,n\}$, and an quantile regression algorithm $\mathcal{A}'$.
\BlankLine
\textbf{Procedure:}\;
\textbf{1.}~Split $\{1, \ldots, n\}$ into two disjoint sets $\mathcal{I}_1$ and $\mathcal{I}_2$. Let $\mathcal{D}_2=\{(x_i,y_{i,1}):i\in\mathcal{I}_2\}$ (WLOG we write $y_{i,1}$ but it can be replaced by $y_{i,j_i}$ where $j_i$ is any index in $1,\ldots,r_i$) and $\mathcal{D}_1=\mathcal{D}_{tr}\backslash \mathcal{D}_2$. If $r_i\ge 2$ for any $i$, then we can, for instance, take $\mathcal{I}_1=\emptyset$ and $\mathcal{I}_2=\{1, \ldots, n\}$.\;

\textbf{2.}~Fit a quantile regression model to $\mathcal{D}_1$:
$$\mathcal{A}' (\mathcal{D}_1) \to (\hat{q}_{\alpha/2}(x),\hat{q}_{1-\alpha/2}(x)).$$

\textbf{3.}~Define
$E_i= \max\{\hat{q}_{\alpha/2} (x_i) - y_{i,1}
, y_{i,1} - \hat{q}_{1-\alpha/2}(x_i)\}$, $i \in  \mathcal{I}_2$, the residuals on $\mathcal{I}_2$. Then compute
$Q'_{1-\alpha}(E, \mathcal{I}_2) := (1 -\alpha)(1 + 1/|\mathcal{I}_2|)$-th empirical quantile of the set $\{E_i: i \in \mathcal{I}_2\}$.
\BlankLine

\Output{The prediction interval at point $\tilde{x}$ is given by
$$[\hat{q}_{\alpha/2} (\tilde{x}) -Q'_{1-\alpha}(E, \mathcal{I}_2), \hat{q}_{1-\alpha/2} (\tilde{x}) +Q'_{1-\alpha}(E, \mathcal{I}_2)].$$}

\end{algorithm}

Despite its generality, conformal prediction mainly focuses on Type-I prediction intervals (\ref{PI2}). To illustrate the limitation of this definition, we assume the regression model outputs the naive prediction $\hat{\mu}\equiv 0$. Then the split conformal prediction interval is given by
$$[-Q_{1-\alpha}(\{|y_i|:i\in \mathcal{I}_2\}, \mathcal{I}_2),Q_{1-\alpha}(\{|y_i|:i\in \mathcal{I}_2\}, \mathcal{I}_2)].$$
We remark that this prediction interval does achieve the coverage guarantee \eqref{PI2}, although typically it is not the prediction interval that we desire. Moreover, we note that conformal prediction does not explicitly consider the conservativeness/width of the prediction intervals. Therefore, one may be interested in the prediction interval with a stronger guarantee, such as the high-quality criterion that we will discuss in Section \ref{sec:DL}. 

\section{DEEP LEARNING FOR PREDICTION INTERVALS} \label{sec:DL}

In this section, we focus on a line of work based on deep learning which explicitly considers the interval width, and thus typically generates width-variable prediction intervals in contrast to conformal prediction. Constructing a high-quality prediction interval requires to balance a tradeoff between the expected interval width and the expected coverage rate maintenance. This viewpoint can be formalized as a constrained stochastic optimization problem, which has been used in previews work such as \cite{khosravi2010lower,pearce2018high,rosenfeld2018discriminative,chen2021learning}. Although it is difficult to solve this problem directly, it provides a framework for developing the following training procedure practically. 


To be more concrete, we aim to find two functions $L$ and $U$, both in a class of functions $\mathcal{H}$ given by a neural network, so that $[L(x),U(x)]$ is the ``optimal-quality'' prediction interval with prediction level $1-\alpha$, which is formulated as the following constrained stochstic optimization problem:
\begin{equation} \label{OP1}
\begin{aligned}
\min_{L,U\in \mathcal{H}\;\text{and}\;L\leq U}\ &\mathbb{E}_{\pi_X}[U(X)-L(X)]\\
\text{subject to } &\mathbb{P}_{\pi}(Y\in[L(X),U(X)]|L,U)\geq 1-\alpha
\end{aligned}   
\end{equation}
where $\mathbb{E}_{\pi_X}$ denotes the expectation with respect to the marginal distribution of $X$. 
Given a set of data $\mathcal{D}=\{(x_i,y_{i,j}): j=1,\ldots,r_i, i=1,\ldots,n\}$, we approximate \eqref{OP1} with the following empirical constrained optimization problem:
\begin{equation}\label{OP2}
\begin{aligned}
\widehat{\opt} (t):
\min_{L,U\in \mathcal{H}\;\text{and}\;L\leq U}\ &\mathbb{E}_{\hat\pi_X}[U(X)-L(X)] \\
\text{subject to }&\mathbb{P}_{\hat\pi}(Y\in[L(X),U(X)]|L,U)\geq 1-\alpha+t
\end{aligned}
\end{equation}
parameterized by $t\in [0,\alpha]$, where $\mathbb{E}_{\hat\pi_X}$, $\mathbb{P}_{\hat\pi}$ are expectation and probability with respect to the empirical distribution constructed from the data $\mathcal{D}$, i.e.,
$\mathbb{E}_{\hat\pi_X}[U(X)-L(X)]=\frac{1}{n}\sum_{i=1}^n(U(x_i)-L(x_i))$,  $\mathbb{P}_{\hat\pi}(Y\in[L(X),U(X)]|L,U)=\frac{1}{n}\sum_{i=1}^n \frac{1}{r_i} \sum_{j=1}^{r_i} \textbf{1}_{y_{i,j}\in[L(x_i),U(x_i)]}$.
The ``safety margin'' $t$ in \eqref{OP2} is introduced to boost the performance guarantee between \eqref{OP1} and \eqref{OP2}.
Note that if $t$ is too small, say, $t=0$, then 
because of finite-sample errors, the true coverage can be smaller than the empirical coverage with substantial probability (in fact, central limit theorem tells us that this happens with probability $\approx1/2$), so that even when the empirical coverage reaches the target level $1-\alpha$, the true coverage can be (much) lower. 
On the other hand, if $t$ is too large, it will greatly reduce the feasible set of (\ref{OP2}), leading to a deterioration in the objective interval width. Therefore the margin $t$ needs to be properly calibrated.

In terms of theory, the learning guarantee of feasibility and optimality between (\ref{OP1}) and (\ref{OP2}) has been extensively studied in \cite{chen2021learning}. For instance, if the class of function $\mathcal{H}$ has finite VC dimension $\vc{\mathcal{H}}$ such as a neural network, then Theorem 4.3 in \cite{chen2021learning} reveals that, after ignoring logarithmic factors, the dataset size $n$ needed to learn a good prediction interval with guaranteed coverage from $\widehat{\opt}(t)$ is of order $\Omega(\vc{\mathcal{H}})$, if $t$ of order $O\big(\sqrt{\vc{\mathcal{H}}/n}\big)$ is adopted. The corresponding optimality gap in width is $O\big(\sqrt{\vc{\mathcal{H}}/n}\big)$.


It is almost impossible to directly solve \eqref{OP2} with a general function set $\mathcal{H}$. Therefore we consider an empirical approach using deep learning. First we derive a Lagrangian formulation of \eqref{OP2}. This formulation introduces the dual multiplier $\lambda$ which can be viewed as a tunable hyper-parameter to balance the tradeoff between the objective and the constraint in \eqref{OP2}. Specifically, we consider
$$\mathcal{L}(\lambda)=\mathbb{E}_{\hat\pi_X}[U(X)-L(X)]+ \lambda (1-\alpha+t- \mathbb{P}_{\hat\pi}(Y\in[L(X),U(X)]|L,U))$$
or
$$\mathcal{L}(\lambda)=\frac{1}{n}\sum_{i=1}^{n}(U(x_i)-L(x_i))+  \frac{\lambda}{n}\sum_{i=1}^{n}  \frac{1}{r_i}\sum_{j=1}^{r_i} \textbf{1}_{y_{i,j}\notin[L(x_i),U(x_i)]}+ \text{constant}$$
where $\lambda$ is the multiplier. Intuitively, if $\lambda$ is large, $(U(x_i)-L(x_i))$ contributes less to the overall loss function, and hence the resulting interval tends to be wide but has a high coverage rate. On the contrary, a small $\lambda$ entails a narrow interval with a low coverage rate. We treat $\lambda$ as a hyper-parameter in the loss function and use deep learning to approximately solve this problem. To achieve this, we assume that $L$ and $U$ are two output neurons given by a neural network model with network parameters $\theta$ to be optimized based on the loss function:
$$\mathcal{L}(\theta;\lambda)=\frac{1}{n}\sum_{i=1}^{n}(U_\theta(x_i)-L_\theta(x_i))+  \frac{\lambda}{n}\sum_{i=1}^{n}   \frac{1}{r_i}\sum_{j=1}^{r_i} \textbf{1}_{y_{i,j}\notin[L_\theta(x_i),U_\theta(x_i)]}$$
Then we run gradient descent on $L(\theta;\lambda)$ with respect to $\theta$ to find the optimal network parameters $\theta^*$. Note that this $L$ is non-smooth with respect to $\theta$. Therefore to implement gradient descent, we use a ``soft'' version of the Lagrangian function. For instance, we can adopt the following form of ``soft'' loss (or alternative losses such as in  \cite{chen2021learning,khosravi2010lower,khosravi2011comprehensive,pearce2018high}):
\begin{equation} \label{equ:loss}
l(\theta):=\frac{1}{n}\sum_{i=1}^{n}(U_\theta(x_i)-L_\theta(x_i)) + \frac{\lambda}{n}\sum_{i=1}^{n} \frac{1}{r_i}\sum_{j=1}^{r_i} \Big(1- \sigma(C_0*(U_\theta(x_i)-y_{i,j}))*\sigma(C_0*(y_{i,j}-L_\theta(x_i)))\Big)
\end{equation}
where $\sigma(t):=\frac{1}{1+e^{-t}}$ is the Sigmoid function and $C_0$ is a constant.
In practice, a mini-batch gradient descent \cite{li2014efficient} is adopted to accelerate the training procedure.
Finally, we note that $\lambda$ is a hyper-parameter during the training procedure which is directly related to the coverage constraint (\ref{OP1}). If we arbitrarily set the value of $\lambda$, the prediction interval output by the networks is not guaranteed to achieve the target prediction level. 

To conclude, the advantage of deep learning is computationally scalable, efficient and easy-to-implement without matrix computation. Moreover, it targets the optimization problem \eqref{OP1} which produces width-variable prediction intervals compared with conformal prediction. On the other hand, its theoretical guarantee cannot be readily derived by its training procedure, and a validation procedure is needed for the performance guarantee of prediction intervals. 

\section{VALIDATION PROCEDURE}\label{sec:calibration}

In this section, we aim to provide a validation algorithm to select the prediction interval model. In brief, our proposed method selects the margin $t$ in \eqref{OP2} in a data-driven manner to guarantee that the selected prediction interval will attain the target prediction level with high confidence. This validation procedure is inspired by \cite{chen2021learning} and \cite{lam2019combating}, but with a generalization to handle possibly multiple simulation runs per design point.

In the standard learning routine, the best hyper-parameter or model is chosen via a validation procedure. To do this, we first train multiple ``candidate'' models, which are allowed to be built by deep learning or any other algorithms. Then we evaluate these trained models on a validation set to select the optimal one among them. 
It is in general a non-trivial task how to properly select the optimal prediction interval model. A naive, but natural approach is that we choose the model with the shortest average interval width, among candidate models whose empirical coverage rate on the validation set reaches the target prediction level (i.e., we use the criterion \eqref{OP2} with $t=0$ to choose the best model). However, as we have discussed in Section \ref{sec:DL}, this natural approach cannot guarantee the feasibility of this model. Therefore a more elaborate validation procedure is needed.




Consider the setting of simulation metamodeling as described in Section \ref{sec:CP} where the observed simulation data are $\{(x_i,y_{i,1}),\ldots,(x_i,y_{i,r}):i=1,\ldots,n\}$. Here we assume the number of simulation replications $r$ at each design point $x_i$ is the same (which generally holds in practice and is adopted in previous work such as \cite{chen2013building}). 
We randomly divide the data into two disjoint subsets: training set and validation set. The training set is used to train multiple candidate prediction interval models, say $\{\PI_j(x)=[L_j(x),U_j(x)]:j=1,\ldots,m\}$. These models can be obtained by setting different values at $\lambda$ in the loss formulation \eqref{equ:loss} in Section \ref{sec:DL}, but can also be a more general collection of models. Then we use the validation data, say $\mathcal{D}_v=\{(x'_i,y'_{i,1}),\ldots,(x'_i,y'_{i,r}):i=1,\ldots,n_v\}$, independent of the training, for the validation procedure. Our goal is to output $K$ models, each with given prediction level $1-\alpha_k$ ($k=1,\ldots,K$).

\begin{algorithm}
\caption{Normalized Prediction Interval Validation}
\label{calibration:normalized}
\DontPrintSemicolon
\SetKwInOut{Input}{Input}\SetKwInOut{Output}{Output}

\textbf{Input:} Candidate models $\{\PI_j=[L_j,U_j]:j=1,\ldots,m\}$, target prediction levels $\{1-\alpha_k\in(0,1):k=1,\ldots,K\}$, validation data $\mathcal{D}_v=\{(x'_i,y'_{i,1}),\ldots,(x'_i,y'_{i,r}):i=1,\ldots,n_v\}$, and confidence level $1-\beta\in(0,1)$.
\BlankLine
\textbf{Procedure:}\;

\textbf{1.}~For each $\PI_j$, $j=1,\ldots,m$, compute its empirical coverage rate on $\mathcal{D}_v$, $\hat{\CR}(\PI_j):=\frac{1}{n_v}\sum_{i=1}^{n_v}\frac{1}{r}\sum_{l=1}^{r}I_{y'_{i,l}\in \PI_j(x'_i)}$. Compute the sample covariance matrix $\hat \Sigma\in \R^{m\times m}$ with $\hat \Sigma_{j_1,j_2} = \frac{1}{n_v}\sum_{i=1}^{n_v}\big(\frac{1}{r}\sum_{l=1}^{r}I_{y'_{i,l}\in \PI_{j_1}(x'_i)}-\hat{\CR}(\PI_{j_1})\big)\big(\frac{1}{r}\sum_{l=1}^{r}I_{y'_{i,l}\in \PI_{j_2}(x'_i)}-\hat{\CR}(\PI_{j_2})\big)$.\;

\textbf{2.}~Let $\hat{\sigma}^2_j=\hat\Sigma_{j,j}$, and compute $q_{1-\beta}$, the $(1-\beta)$-quantile of $\max_{1\leq j\leq m}\{Z_j/\hat{\sigma}_j:\hat{\sigma}_j>0\}$ where $(Z_1,\ldots,Z_m)$ is a multivariate Gaussian with mean zero and covariance $\hat \Sigma$.\;

\textbf{3.}~For each target level $1-\alpha_k$, $k=1,\ldots,K$, compute
\begin{equation*}
\begin{aligned}
j^*_{1-\alpha_k} = \argmin{1\leq j\leq m}\Big\{\frac{1}{n_v}\Big(\sum_{i=1}^{n_v}\lvert U_j(x'_i)-L_j(x'_i)\rvert\Big):\hat{\CR}(\PI_j)\geq 1-\alpha_k + \frac{ q_{1-\beta}\hat{\sigma}_j}{\sqrt{n_v}}\Big\}. 
\end{aligned}
\end{equation*}

\BlankLine
\Output{$\PI_{j^*_{1-\alpha_k}}$ for $k=1,\ldots,K$.}
\end{algorithm}

\begin{algorithm}
\caption{Unnormalized Prediction Interval Validation}
\label{calibration:unnormalized}
\DontPrintSemicolon
\SetKwInOut{Input}{Input}\SetKwInOut{Output}{Output}

\textbf{Input:} Same as in Algorithm \ref{calibration:normalized}.

\BlankLine

\textbf{Procedure:}\;

\textbf{1.}~Same as in Algorithm \ref{calibration:normalized}.\;

\textbf{2.}~Compute $q'_{1-\beta}$, the $(1-\beta)$-quantile of $\max\{Z_j:1\leq j\leq m\}$ where $(Z_1,\ldots,Z_m)$ is a multivariate Gaussian with mean zero and covariance $\hat\Sigma$.\;

\textbf{3.}~For each target level $1-\alpha_k$, $k=1,\ldots,K$, compute
\begin{equation*}
j^*_{1-\alpha_k} = \argmin{1\leq j\leq m}\Big\{\frac{1}{n_v}\Big(
\sum_{i=1}^{n_v}\lvert U_j(x'_i)-L_j(x'_i)\vert\Big):\hat{\CR}(\PI_j)\geq 1-\alpha_k + \frac{ q'_{1-\beta}}{\sqrt{n_v}}\Big\}.
\end{equation*}
\BlankLine
\Output{$\PI_{j^*_{1-\alpha_k}}$ for $k=1,\ldots,K$.}
\end{algorithm}


Algorithm \ref{calibration:normalized} describes our proposed validation procedure. In this algorithm, we check the feasibility of each candidate model on the validation set, using the criterion $\hat{\CR}(\PI_j):=\frac{1}{n_v}\sum_{i=1}^{n_v}\frac{1}{r}\sum_{j=1}^{r}I_{y'_{i,j}\in \PI_j(x'_i)}\geq 1-\alpha_k+\epsilon_j$ for some selected margins $\epsilon_j$. Then we choose the one among them with the smallest average interval width. The choice of $\epsilon_j$ is based on the quantile of the supremum of a multivariate Gaussian distribution. We give some intuitive explanations of this algorithm here. Let $\CR(\PI_j):=\mathbb{P}_{\pi}(Y\in \PI_j(X)|\PI_j)$ denote the expected coverage rate of $\PI_j$. Then the multivariate central limit theorem implies that $$\sqrt{n_v}\big(\hat{\CR}(\PI_1)-\CR(\PI_1),\ldots,\hat{\CR}(\PI_m)-\CR(\PI_m)\big)\overset{d}{\to} N(0,\Sigma)$$
where $\Sigma$ is the covariance matrix. Approximating $\Sigma$ with the sample covariance $\hat\Sigma$ from Step 1 of Algorithm \ref{calibration:normalized} and applying the continuous mapping theorem, we have $$\sqrt{n_v}\max_{j}(\hat{\CR}(\PI_j)-\CR(\PI_j))/\hat{\sigma}_j\overset{d}{\to} \max_j Z_j/\hat{\sigma}_j$$ where $(Z_j)_{j=1,\ldots,m}$  follows $N(0,\hat\Sigma)$. By using the $1-\beta$ quantile of $\max_j Z_j/\hat{\sigma}_j$ in the margins, we should expect that
$\CR(\PI_j)\geq \hat{\CR}(\PI_j)-q_{1-\beta}\hat\sigma_j/\sqrt{n_v}$
for all $j=1,\ldots,m$ uniformly with probability $\approx 1-\beta$.

In fact, we have the following theorem to precisely describe the error bound based on recent high-dimensional Berry-Esseen theorems \cite{chernozhukov2017central}. This theorem is a generalization of the result in \cite{chen2021learning} to handle multiple simulation runs $r$ per design point.
\begin{theorem}[Feasibility Guarantee]\label{feasibility:normalized validator_simple}
Let $1-\underline{\alpha}:=\max_{j=1,\ldots,m}\CR(\PI_j)$, $1-\alpha_{\min}:=1-\min_{k=1,\ldots,K}\alpha_k$, and $\tilde{\alpha}:=\min\{\alpha_{\min},1-\max_{k=1,\ldots,K}\alpha_k\}$. For every collection of interval models $\{\PI_j:1\leq j\leq m\}$, every $n_v$, and $\beta\in(0,\frac{1}{2})$, the prediction intervals output by Algorithm \ref{calibration:normalized} satisfy
\begin{eqnarray}
\notag &&\mathbb P_{\mathcal D_v}(\CR(\PI_{j^*_{1-\alpha_k}})\geq 1-\alpha_k\; \text{for all}\; k=1,\ldots,K)\\
\notag &\geq& 1-\beta-C\bigg(\Big(\frac{r\log^7(mn_v)}{n_v\tilde{\alpha}}\Big)^{\frac{1}{6}}+
\exp\big(-C_2n_v\min\big\{\epsilon,\frac{\epsilon^2}{\underline{\alpha}(1-\underline{\alpha})}\big\}\big)\bigg)\label{finite sample error:normalized_simple}
\end{eqnarray}
with $\epsilon=\max\{\alpha_{\min}-\underline{\alpha}-C_1\sqrt{\frac{\log(m/\beta)}{n_v}},0\}$, where $\CR(\PI_j):=\mathbb{P}_{\pi}(Y\in \PI_j(X)|\PI_j)$, $\mathbb P_{\mathcal D_v}$ denotes the probability with respect to the validation data $\mathcal D_v$, and $C,C_1,C_2$ are universal constants.
\end{theorem}

The proof of Theorem \ref{feasibility:normalized validator_simple} is given in Appendix \ref{sec:proofs}. Theorem \ref{feasibility:normalized validator_simple} shows that the prediction interval output by our validation procedure, Algorithm \ref{calibration:normalized}, has finite-sample guarantee on the coverage rate attainment.
In general any prediction intervals that satisfy the constraint in Step 3 have such a feasibility guarantee. Meanwhile, in order to have better performance on the objective, we select the one among them with the smallest average interval width on the validation set. 
In addition to Theorem \ref{feasibility:normalized validator_simple}, our validation procedure also possesses guaranteed performance regarding the optimality of the expected interval width among candidate models, which is detailed as follows: 

\begin{theorem}[Optimality Guarantee] \label{optimality guarantee: normalized}
Assume all the candidate prediction intervals in Algorithm \ref{calibration:normalized} are selected from a hypothesis class $\mathcal H$ whose envelope $H(x):=\sup_{h\in \mathcal H}\lvert h(x) - \mathbb E_{\pi_X}[h(X)] \rvert$ has a finite sub-Gaussian norm $\Vert H \Vert_{\psi_2}<\infty$. For every $\eta>0$ we have
\begin{eqnarray*}
&&\mathbb P_{\mathcal D_v}\Big(\mathbb E_{\pi_X}[U_{j^*_{1-\alpha_k}}(X) - L_{j^*_{1-\alpha_k}}(X)]\leq \min_{j:\mathrm{CR(\mathrm{PI}_j)\geq 1-\alpha_k+\eta}}\mathbb E_{\pi_X}[U_{j}(X) - L_{j}(X)]+2C\eta\Vert H\Vert_{\psi_2}\text{ for all }k=1,\ldots,K\Big)\\
&\geq&1-8m\exp \Big( -\frac{1}{4}\max\big\{\eta - C\sqrt{\frac{\log (m/\beta)}{n_v}}, 0\big\}^2n_v\Big)
\end{eqnarray*}
for some universal constant $C$.
\end{theorem}

The proof of Theorem \ref{optimality guarantee: normalized} is given in Appendix \ref{sec:proofs}. Theorem \ref{optimality guarantee: normalized} shows that the prediction interval output by our validation procedure, Algorithm \ref{calibration:normalized}, is nearly-optimal among those candidate models.
A similar validation strategy, viewed as a ``unnormalized'' (as opposed to ``normalized'') version of Algorithm \ref{calibration:normalized} when we handle the variance term $\hat\sigma_j$, is described in Algorithm \ref{calibration:unnormalized}. Its performance guarantee can be established similarly as in Theorems \ref{feasibility:normalized validator_simple} and \ref{optimality guarantee: normalized}.

\section{EXPERIMENTS}\label{sec:experiments}

In this section, we conduct some numerical experiments. 
We assume the experiment design has been fixed and the simulation data have been supplied in a single batch in a random or arbitrary manner. Our goal is to use the given data to build a prediction interval for future simulation outputs.

We describe the baseline algorithms for comparisons. For SK, we use Lemma \ref{SK} to generate prediction intervals under the posterior distribution. Our estimation of the GP structure follows the instruction in Section 2.2 and Section 4 in \cite{ankenman2010stochastic}. For split conformal prediction (SCP), our implementation is based on the description in Section \ref{sec:CP} where the base regression algorithm is a neural network consisting of 1 hidden layer with $20$ neurons using mean squared loss.  
For quantile regression forests (QRF), this is the well-known algorithm proposed in \cite{meinshausen2006quantile}. For split conformalized quantile regression (SCQR), our implementation is based on the description in Section \ref{sec:CP} where the base quantile regression algorithm is exactly the QRF. For neural-network-based prediction intervals, we adopt the loss function and method described in Section \ref{sec:DL}. The base neural network consists of 1 hidden layer with $20$ neurons and the output layer with 2 neurons representing $L$, $U$. By adjusting $\lambda$ in \eqref{equ:loss}, prediction interval models with different coverage rates can be trained, which are then used in our validation algorithms to obtain the final model. We implement three validation strategies: Algorithm \ref{calibration:normalized} (NNGN), Algorithm \ref{calibration:unnormalized} (NNGU), 
and the natural validation scheme (NNVA). In NNVA, we directly compares the empirical coverage rates on the validation set to the target levels without the Gaussian margins, as we mentioned in Section \ref{sec:calibration}. The validation data are split from the training data:  60\% data for training multiple models and 40\% data for validation procedure.

We apply the above methods to construct prediction intervals with prediction level $1-\alpha=95\%$. Each experiment is repeated for $N=50$ times to estimate the confidence of coverage attainment.  The performance of each method is evaluated based on two metrics: exceedance probability ($EP$) and interval width ($IW$). $EP$ captures the success ratio in achieving the target prediction level among $N=50$ experimental repetitions (which is viewed as a confidence score), while $IW$ indicates the average interval width among $N=50$ experimental repetitions. Formally:
$$EP=\frac{1}{N}\sum_{i = 1}^{N}\textbf{1}\{CR_{i} \geq 1-\alpha\}, \ \ IW=\frac{1}{Nn_{te}}\sum_{i=1}^{N}\sum_{j =1}^{n_{te}}(U_{i}(x_j)-L_{i}(x_j))
$$
where $n_{te}$ is the size of test points, $CR_i$ is the empirical coverage rate on test data of the $i$-th prediction interval generated from the $i$-th repetition and the target prediction level $1-\alpha=95\%$. Throughout our experiments, the confidence level $1-\beta$ is set to $0.95$ in the validation procedure. The best result is achieved by the model with the smallest $IW$ value among those with $EP \ge 0.95$. If no model achieves $EP \ge 0.95$, then the one with highest $EP$ is the best. 

\subsection{M/M/1 Queue}
We conduct experiments on an M/M/1 queue example as considered in \cite{ankenman2010stochastic}. Let $Y(x)$ be the steady-state number of customers in an M/M/1 queue with service rate 1 and arrival rate $x$. Our goal is to model $Y(x)$. It is well known that $Y(x)$ follows a geometric distribution with mean $\mathbb{E}[Y(x)] = \frac{x}{1-x}$ but we assume we do not have access to this information in our experiments.

We consider two experiment designs: (1) Sparse and deep design: Training data consist of 7 design points, $x = 0.3, 0.4, 0.5,\ldots, 0.9$, making $50$ simulation replications at each of them. (2) Dense and shallow design: Training data consist of 31 design points, $x = 0.30, 0.32, 0.34,\ldots, 0.90$, making $5$ simulation replications at each of them. 
In both cases, we implement different methods on these training data to construct prediction intervals and then evaluate the performance of each method on new test data, which are obtained at 50 input points i.i.d. drawn from $\text{Uniform}[0.3, 0.9]$, making $100$ simulation replications at each of them. We use sufficient test data so that the test coverage can represent the population coverage in \eqref{PI1}. We aim to build prediction intervals with 95\% prediction level for the simulation output $Y(x)$ over the domain $0.3\le x\le 0.9$, which should cover at least $95\%$ test data while maintaining a small interval width. The numerical results are shown in Table \ref{single-level} which reports the values of $EP$ and $IW$ for each method under 2 experimental designs. 

\begin{table*}
  \centering
  \renewcommand{\arraystretch}{1}
  \setlength{\tabcolsep}{2pt}
  \begin{tabular}{|c|cc|cc|cc|cc|cc|cc|cc|}
    \hline
&\multicolumn{2}{c|}{SK}& \multicolumn{2}{c|}{SCP}& \multicolumn{2}{c|}{QRF} &\multicolumn{2}{c|}{SCQR} &\multicolumn{2}{c|}{NNVA}&\multicolumn{2}{c|}{NNGN}&\multicolumn{2}{c|}{NNGU}\\
Data &EP & IW& EP & IW& EP & IW &EP & IW &EP & IW & EP & IW & EP & IW\\
\hline
Experiment Design 1 & 1.00 & 9.84 & 0.70 & 8.89 & 0.88 & 4.14 & 1.00 & 5.80 & 0.66 & 3.98 & \textbf{0.98} & \textbf{4.75} & 1.00 & 4.88\\
Experiment Design 2 & 1.00 & 10.46 & 0.84 & 9.61 & 0.74 & 3.92 & 1.00 & 6.75 & 0.62 & 4.04  & \textbf{0.96} & \textbf{4.59} & 0.96 & 4.84 \\
\hline
 \end{tabular}
\caption{Prediction intervals with 95\% prediction level for the M/M/1 queue. The best results are in \textbf{bold}.}
\label{single-level}
\end{table*}

\subsection{Computer Communication Network}

We conduct experiments on a computer communication network example borrowed from \cite{lam2021subsampling,lin2015single}.  Figure \ref{fig1} illustrates the structure of this computer communication network. It consists of 4 message-processing units (nodes) which are connected by 4 transport channels (edges). There are external messages that will enter into the network. We assume that the length of the external messages are i.i.d. and follow an exponential distribution with a certain mean, denoted as $x$. For every pair $(i, j)$ of nodes ($i \ne j$), 
external messages arriving to node $i$ that are to be transmitted to node $j$ follow a Poisson arrival process with rate $\lambda_{i,j}$ where their transmission path is a fixed and known. The values for $\lambda_{i,j}$’s are given in Table \ref{arrivalrates} and assumed to be known. Suppose that each node takes a constant time of $0.001$ seconds to process a message with unlimited storage capacity, and each edge has a capacity of 275000 bits. All messages transmit through the edges with a constant velocity of 150000 miles per second, and the $i$-th edge has
a length of $100 i$ miles for $i = 1, 2, 3, 4$. Therefore, the total time that a message of length $l$ bits occupies the $i$-th edge is
$\frac{l}{275000} +
\frac{100i}{150000}$ seconds. Suppose that the computer network is empty at time zero. The simulation random output of
interest is the delay of the first 30 messages that arrive to the system, or $Y=\frac{1}{30}\sum_{i=1}^{30} T_k$, where $T_k$ is the time that the $k$-th message takes to be transmitted from its entering node to destination node. Obviously, this value depends on the mean of the exponential distribution that the length of messages follow, which is $x$. Our goal is to model $Y(x)$. Unlike the M/M/1 queue example, the analytical form of the conditional distribution of $Y(x)$ in this example cannot be easily obtained.

\begin{figure}
    \centering
    \includegraphics[width=0.7\textwidth]{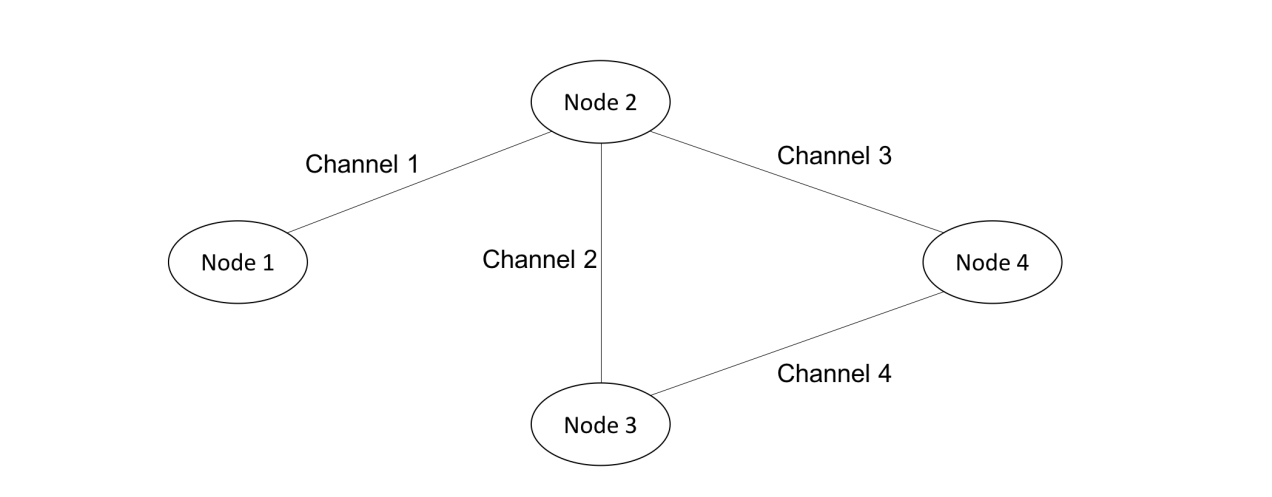}
    \caption{A computer communication network with four nodes and four channels.}
    \label{fig1}
\end{figure}

\begin{table*}
\begin{tabular}{|c|cccc|}
\hline
\diagbox{node $i$}{node $j$} & 1 & 2 & 3 & 4 \\
 \hline
 1 & n.a. &  40 & 30 & 35 \\
 2 & 50 & n.a. & 45 & 15 \\
 3 & 60 & 15 & n.a. & 20 \\
 4 & 25 & 30 & 40 & n.a. \\
 \hline
\end{tabular}
\caption{Arrival rates $\lambda_{i,j}$ of messages to be transmitted from node $i$ to node $j$.}
\label{arrivalrates}
\end{table*}

Like the previous experiment, we consider two designs: (1) Sparse and deep design: Training data consist of 11 design points, $x = 200, 240, 280, \ldots, 600$, making $50$ simulation replications at each of them. (2) Dense and shallow design: Training data consist of 41 design points, $x = 200, 210, 220, \ldots, 600$, making $5$ simulation replications at each of them. 
In both cases, we implement different methods on these training data to construct prediction intervals and then evaluate the performance of each method on new test data, which are obtained at 50 input points i.i.d. drawn from $\text{Uniform}[200, 600]$, making $100$ simulation replications at each of them. We aim to build prediction intervals with 95\% prediction level for the simulation output $Y(x)$ over the domain $200\le x\le 600$, which should cover at least $95\%$ test data while maintaining a small interval width. The numerical results are shown in Table \ref{single-level2} which reports the values of $EP$ and $IW$ for each method under 2 experimental designs. 

\begin{table*}
  \centering
  \renewcommand{\arraystretch}{1}
  \setlength{\tabcolsep}{2pt}
  \begin{tabular}{|c|cc|cc|cc|cc|cc|cc|cc|}
    \hline
&\multicolumn{2}{c|}{SK}& \multicolumn{2}{c|}{SCP}& \multicolumn{2}{c|}{QRF} &\multicolumn{2}{c|}{SCQR} &\multicolumn{2}{c|}{NNVA}&\multicolumn{2}{c|}{NNGN}&\multicolumn{2}{c|}{NNGU}\\
Data &EP & IW& EP & IW& EP & IW &EP & IW &EP & IW & EP & IW & EP & IW\\
\hline
Experiment Design 1 & 1.00 & 8.14 & 0.88 & 8.84 & 0.84 & 6.36 & 0.88 & 8.20 & 0.80 & 6.05 & \textbf{1.00} & \textbf{6.52} & 1.00 & 6.56\\
Experiment Design 2 & 1.00 & 7.93 & 0.78 & 7.67 & 0.72 & 6.03 & 0.70 & 6.18 & 0.52 & 5.82 & \textbf{0.98} & \textbf{6.32} & 0.98 & 6.34 \\
\hline
\end{tabular}
\caption{Prediction intervals with 95\% prediction level for the computer communication network. The best results are in \textbf{bold}. The magnitude of IW results are all $10^{-3}$.}
\label{single-level2}
\end{table*}

\subsection{Results and Discussion}
From our results in Tables \ref{single-level} and \ref{single-level2}, we have the following observations on the performance of each approach. Prediction intervals generated by SK tend to be conservative/wide while they typically achieve the desired prediction level. Prediction intervals generated by QRF tend to be narrow although they are not guaranteed to achieve the desired prediction level with high confidence. SCP and SCQR do not perform well in our experiments probably because they mainly aim to generate Type-I prediction intervals.
Among the methods which achieve the target prediction level with high confidence ($EP\ge 0.95$), the IW values of the NNGN are the smallest in all experiments. In addition, the NNGN and NNGU produce similar competitive results, achieving the prediction level with high confidence ($EP\ge 0.95$) while maintaining small interval width. In contrast, the naive validation scheme NNVA cannot attain the desired prediction level satisfactorily as their $EP\le 0.8$ in all experiments, which demonstrates the effectiveness of our validation algorithms. This observation coincides with our discussion in Section \ref{sec:DL}: The margin parameter $t$ in \eqref{OP2} needs to be properly calibrated and a naive choice of $t$ will fail.

In addition, we discuss the running time of implementing each approach. SK becomes slower when more design points are placed: In this case, the running time of the matrix inversion in SK increases (at the order of $O(n^3)$) and the evaluation at test data costs more computational time. The running times of SCP and SCQR are essentially determined by the running time of training the regression/quantile regression model which is specified by the user. QRF is fast and computationally scalable in the sense that its running time is not sensitive to the sample size. NN-based methods are also computationally scalable and their evaluation at test data is extremely fast compared with the training time. The running time of our proposed validation procedures is almost negligible. Therefore, we can expect that in the situation where a lot of design points should be placed, the computational time of NN-based methods (including validation procedures) will be less than SK.

\section{Conclusion}

In this paper, we study an alternative simulation metamodeling approach using output prediction intervals, which produce a bracket of the simulation output values with prescribed coverage level. We cast this task into an empirical constrained optimization framework to train the prediction interval that attains accurate prediction in terms of overall coverage and interval width. We describe an approach to set up an empirical loss to approximately solve the constrained optimization problem, and design a training procedure for the lower and upper bounds of the prediction interval represented as neural networks. We further present a validation machinery to optimally select the best hyperparameter in the formulated loss and show that it enjoys distribution-free finite-sample guarantees on the prediction performance. In addition, as comparison we also study another approach in constructing prediction intervals by adapting conformal prediction tools into the setting of simulation metamodeling. Our numerical experiments appear to corroborate well our theory and demonstrate the effectiveness of our proposed approaches.

\begin{acks}
We gratefully acknowledge support from the National Science Foundation under grants CAREER CMMI-1834710 and IIS-1849280 and the JP Morgan Chase Faculty Research Award. A preliminary version of this work has appeared in \cite{lam2021neural}. 
\end{acks}

\bibliographystyle{ACM-Reference-Format}
\bibliography{sample-base}

\appendix

\section{Proofs} \label{sec:proofs}

\begin{proof}[Proof of Lemma \ref{SK2}]
Note that $\bar{\mathbf y}$ and $\tilde{y}$ are jointly Gaussian since they come from the same model \eqref{SK}. To compute the joint finite-dimensional distribution of $\bar{\mathbf y},\tilde{y}$, we note that for $i_1\ne i_2$,
\begin{align*}
\text{Cov}[y_{i_1,j_1},y_{i_2,j_2}] &= \text{Cov}[M(x_{i_1}) +\varepsilon_{j_1}(x_{i_1}), M(x_{i_2}) + \varepsilon_{j_2}(x_{i_2})]\\
&= \text{Cov}[M(x_{i_1}), M(x_{i_2})] + \text{Cov}[\varepsilon_{j_1}(x_{i_1}), \varepsilon_{j_2}(x_{i_2})] \quad \text{since } M \text{ and } \varepsilon \text{ are independent}\\
&= \sigma^2(x_{i_1},x_{i_2}) + c(x_{i_1},x_{i_2})   
\end{align*}
and thus
$$\text{Cov}[\bar{y}_{i_1},\bar{y}_{i_2}] = \frac{1}{r_{i_1}} \frac{1}{r_{i_2}} \sum_{j_1=1}^{r_{i_1}} \sum_{j_2=1}^{r_{i_2}}\text{Cov}[y_{i_1,j_1},y_{i_2,j_2}] = \sigma^2(x_{i_1},x_{i_2}) + c(x_{i_1},x_{i_2}).$$
For $i_1= i_2$, we have
\begin{align*}
\text{Cov}[y_{i_1,j_1},y_{i_1,j_2}] &= \text{Cov}[M(x_{i_1}) +\varepsilon_{j_1}(x_{i_1}), M(x_{i_1}) + \varepsilon_{j_2}(x_{i_1})]\\
&= \text{Cov}[M(x_{i_1}), M(x_{i_1})] + \text{Cov}[\varepsilon_{j_1}(x_{i_1}), \varepsilon_{j_2}(x_{i_1})] \quad \text{since } M \text{ and } \varepsilon \text{ are independent}\\
&= \sigma^2(x_{i_1},x_{i_1}) + c(x_{i_1},x_{i_1})\textbf{1}_{(j_1=j_2)}  \quad \text{since } \varepsilon_1(x_{i_1}), \varepsilon_2(x_{i_1}), \ldots \text{ are independent}\\
\end{align*}
and thus
$$\text{Cov}[\bar{y}_{i_1},\bar{y}_{i_1}] = \frac{1}{r_{i_1}} \frac{1}{r_{i_1}} \sum_{j_1=1}^{r_{i_1}} \sum_{j_2=1}^{r_{i_1}}\text{Cov}[y_{i_1,j_1},y_{i_1,j_2}] = \sigma^2(x_{i_1},x_{i_2}) + \frac{c(x_{i_1},x_{i_1})}{r_{i_1}}.$$
Similarly, $$\text{Cov}[\tilde{y},y_{i,j}] = \text{Cov}[M(\tilde{x}) +\varepsilon(\tilde{x}), M(x_i) + \varepsilon_j(x_i)] = \text{Cov}[M(\tilde{x}), M(x_i)] + \text{Cov}[\varepsilon(\tilde{x}), \varepsilon_j(x_i)]= \sigma^2(\tilde{x},x_i) + c(\tilde{x},x_i)$$
and thus
$$\text{Cov}[\tilde{y},\bar{y}_{i}] = \frac{1}{r_i}  \sum_{j=1}^{r_i}\text{Cov}[\tilde{y},y_{i,j}] = \sigma^2(\tilde{x},x_i) + c(\tilde{x},x_i).$$
Hence we have
$$\begin{bmatrix} \bar{\mathbf y} \\ \tilde{y} \end{bmatrix} \sim \mathcal{N}\left(\begin{bmatrix} \mu({\mathbf x}) \\ \mu(\tilde{x}) \end{bmatrix},  \begin{bmatrix} \sigma^2({\mathbf x},{\mathbf x})+c({\mathbf x},{\mathbf x}) & \sigma^2({\mathbf x},\tilde{x})+c({\mathbf x},\tilde{x}) \\ \sigma^2(\tilde{x},{\mathbf x})+c(\tilde{x},{\mathbf x}) & \sigma^2(\tilde{x},\tilde{x})+c(\tilde{x},\tilde{x}) \end{bmatrix}\right).$$
Then we can obtain the conditional distribution:
$$p(\tilde{y}|\tilde{x}, X, Y)=\frac{p(\tilde{y}, Y|\tilde{x}, X)}{p(Y |\tilde{x}, X)} \sim \mathcal{N}(\tilde{\mu}(\tilde{x}),\tilde{\sigma}^2(\tilde{x})).$$
\end{proof}

\begin{proof}[Proof of Theorem \ref{thm:SCP}]
We first write the target probablity as
\begin{align*}
&\mathbb{P}(\tilde{y}\in[\hat{\mu}(\tilde{x}) -Q_{1-\alpha}(R, \mathcal{I}_2), \hat{\mu}(\tilde{x}) +Q_{1-\alpha}(R, \mathcal{I}_2)])\\
=&\mathbb{P}(|\tilde{y}-\hat{\mu}(\tilde{x})| \le Q_{1-\alpha}(\{|y_{i,1}-\hat{\mu}(x_i)|: i\in \mathcal{I}_2\}, \mathcal{I}_2))\\
=&\mathbb{P}(\tilde{z} \le Q_{1-\alpha}(\{z_i: i\in \mathcal{I}_2\}, \mathcal{I}_2))
\end{align*}
where $z_i=|y_{i,1}-\hat{\mu}(x_i)|$ and $\tilde{z}=|\tilde{y}-\hat{\mu}(\tilde{x})|$, and they are i.i.d. by the construction of Algorithm \ref{algo:SCP}. The last expectation is taken with respect to $\{z_i: i\in \mathcal{I}_2\}\cup \{\tilde{z}\}$. Let $z_{(k,|\mathcal{I}_2|)}$ denotes the $k$-th smallest value in $\{z_i: i\in \mathcal{I}_2\}$ and $z_{(k,|\mathcal{I}_2|+1)}$ denotes the $k$-th smallest value in $\{z_i: i\in \mathcal{I}_2\}\cup \{\tilde{z}\}$. Then 
$$Q_{1-\alpha}(\{z_i: i\in \mathcal{I}_2\}, \mathcal{I}_2))=z_{(\lceil \bar{\alpha} \rceil, |\mathcal{I}_2|)}$$
where $\bar{\alpha}=(1 -\alpha)(|\mathcal{I}_2|+1)$. We note that if $\tilde{z} \le z_{(\lceil \bar{\alpha} \rceil,|\mathcal{I}_2|)}$, then $z_{(\lceil \bar{\alpha} \rceil,|\mathcal{I}_2|+1)}$ is the larger of $z_{(\lceil \bar{\alpha} \rceil-1,|\mathcal{I}_2|+1)}$ and $\tilde{z}$, and thus $
\tilde{z}\le z_{(\lceil \bar{\alpha} \rceil,|\mathcal{I}_2|+1)}$. Conversely, if $
\tilde{z}\le z_{(\lceil \bar{\alpha} \rceil,|\mathcal{I}_2|+1)}$, then since $z_{(\lceil \bar{\alpha} \rceil,|\mathcal{I}_2|+1)} \le z_{(\lceil \bar{\alpha} \rceil,|\mathcal{I}_2|)}$, we have $\tilde{z}\le z_{(\lceil \bar{\alpha} \rceil,|\mathcal{I}_2|)}$. Therefore, we have the equivalence
$
\{\tilde{z}\le z_{(\lceil \bar{\alpha} \rceil,|\mathcal{I}_2|+1)}\} = \{\tilde{z}\le z_{(\lceil \bar{\alpha} \rceil,|\mathcal{I}_2|)}\}$ and thus
\begin{align*}
&\mathbb{P}(\tilde{z} \le Q_{1-\alpha}(\{z_i: i\in \mathcal{I}_2\}, \mathcal{I}_2))\\
=& \mathbb{P}(\tilde{z}\le z_{(\lceil \bar{\alpha} \rceil,|\mathcal{I}_2|)})\\
=& \mathbb{P}(\tilde{z}\le z_{(\lceil \bar{\alpha} \rceil,|\mathcal{I}_2|+1)})\\
=& \frac{1}{|\mathcal{I}_2|+1}\sum_{z\in \{z_i: i\in \mathcal{I}_2\}\cup \{\tilde{z}\}}\mathbb{P}(z\le z_{(\lceil \bar{\alpha} \rceil,|\mathcal{I}_2|+1)}) \ \text{ by symmetry of $\{z_i: i\in \mathcal{I}_2\}\cup \{\tilde{z}\}$ }\\
=& \frac{1}{|\mathcal{I}_2|+1} \mathbb{E}\left[ \sum_{z\in \{z_i: i\in \mathcal{I}_2\}\cup \{\tilde{z}\}} \textbf{1}_{(z\le z_{(\lceil \bar{\alpha} \rceil,|\mathcal{I}_2|+1)})}\right]\\
\ge & \frac{1}{|\mathcal{I}_2|+1} \mathbb{E}\left[ \lceil \bar{\alpha} \rceil\right]\\
\ge & 1-\alpha.
\end{align*}
\end{proof}

\begin{proof}[Proof of Theorem \ref{feasibility:normalized validator_simple}]
The proof is similar to the proof of Theorem 5.1 in \cite{chen2021learning} but with some modifications. We consider the random vectors
$W_i:=\left(\frac{1}{r}\sum_{l=1}^{r}I_{y'_{i,l}\in \PI_1(x'_i)},\ldots,\frac{1}{r}\sum_{l=1}^{r}I_{y'_{i,l}\in \PI_m(x'_i)}\right)$ where $i=1,\ldots,n_v$.
It is easy to see that $W_i$ are i.i.d. for $i=1,\ldots,n_v$ and \begin{equation} \label{equ:proof}
\begin{split}
&W_i^{(t)}=\frac{1}{r}\sum_{l=1}^{r}I_{y'_{i,l}\in \PI_t(x'_i)}\in [0,1],\ \ \frac{1}{n_v}\sum_{i=1}^{n_v}W_i^{(t)}=\hat{\CR}(\PI_t),\ \ \mathbb{E}[W_i^{(t)}]=\CR(\PI_t),\\
& \frac{1}{r}\Big(\CR(\PI_t)\big(1-\CR(\PI_t)\big)\Big)\le \text{Var}[W_i^{(t)}]\le \CR(\PI_t)\big(1-\CR(\PI_t)\big).
\end{split}
\end{equation}
To see the variance bound of $W_i^{(t)}$ in the last line, we note that
\begin{align*}
\text{Var}[W_i^{(t)}]=& \mathbb{E}[\text{Var}[W_i^{(t)}|x'_i]]+\text{Var}[\mathbb{E}[W_i^{(t)}|x'_i]]\\
=& \mathbb{E}[\frac{1}{r} p_t(x'_i)(1-p_t(x'_i))]+\text{Var}[p_t(x'_i)]\\
=& \frac{1}{r}\Big(\CR(\PI_t)- \mathbb{E}[p_t(x'_i)^2]\Big)+ \mathbb{E}[p_t(x'_i)^2]-(\CR(\PI_t))^2\\
=& \frac{1}{r}\CR(\PI_t)-(\CR(\PI_t))^2 +(1-\frac{1}{r}) \mathbb{E}[p_t(x'_i)^2]
\end{align*}
where $p_t(x'_i)$ is the conditional probability $\mathbb{P}(y\in \PI_t(x'_i)|x'_i)$ and obviously $\mathbb{E}[p_t(x'_i)]=\CR(\PI_t)$. Next we plug the following inequality
$$(\CR(\PI_t))^2=\mathbb{E}[p_t(x'_i)]^2\le \mathbb{E}[p_t(x'_i)^2]\le \mathbb{E}[p_t(x'_i)]=\CR(\PI_t),$$
into the last equation to obtain the the variance bound in \eqref{equ:proof}.

We define the following events:
\begin{align*}
E_1&=\big\{\hat{\CR}(\PI_j)\geq 1-\alpha_{\min} + \frac{ q_{1-\beta}\hat\sigma_j}{\sqrt{n_v}}\text{ for some }j=1,\ldots,m\big\}\\
E_2&=\big\{\CR(\PI_j)\geq\hat{\CR}(\PI_j)-  \frac{ q_{1-\beta}\hat\sigma_j}{\sqrt{n_v}}\text{ for all $j$ such that }\CR(\PI_j)\in(\tilde{\alpha}/2,1-\alpha_{\min})\big\}\\
E_3&=\big\{\hat{\CR}(\PI_j)<\tilde{\alpha} + \frac{ q_{1-\beta}\hat\sigma_j}{\sqrt{n_v}}\text{ for all $j$ such that }\CR(\PI_j)\leq \tilde{\alpha}/2\big\}.
\end{align*}
Then we have 
\begin{align*}
\notag \mathbb P_{\mathcal D_v}(\CR(\PI_{j^*_{1-\alpha_k}})\geq 1-\alpha_k\; \text{for all}\; k=1,\ldots,K)\geq& \mathbb P_{\mathcal D_v}(E_1\cap E_2\cap E_3)\\
\notag\geq&1-\mathbb P_{\mathcal D_v}(E_1^c)-\mathbb P_{\mathcal D_v}(E_2^c)-\mathbb P_{\mathcal D_v}(E_3^c)\\
=&\mathbb P_{\mathcal D_v}(E_2)-\mathbb P_{\mathcal D_v}(E_1^c)-\mathbb P_{\mathcal D_v}(E_3^c).
\end{align*}
Therefore, it is sufficient to derive bounds for the probability of those events.

Now we apply Berry-Esseen theorems (Lemma J.11 in \cite{chen2021learning}) to the first probability $\mathbb P_{\mathcal D_v}(E_2)$
\begin{eqnarray*}
\mathbb P_{\mathcal D_v}(E_2)\geq 1-\beta-C\Big(\Big(\frac{r\log^7(mn_v)}{n_v\tilde{\alpha}}\Big)^{\frac{1}{6}}+\frac{\sqrt{r}\log^2 (mn_v)}{\sqrt{n_v\tilde{\alpha}}}+m\exp\big(-c(n_v\frac{\tilde{\alpha}}{r})^{2/3}\big)\Big)
\end{eqnarray*}
since Assumption 3 and 4 in \cite{chen2021learning} for Berry-Esseen theorems hold with $\eta=\frac{1}{r}\tilde{\alpha}/2\cdot (1-\tilde{\alpha}/2)\geq \frac{1}{4r}\tilde{\alpha}$ and $D=\frac{C\sqrt{r}}{\sqrt{\tilde{\alpha}}}$ by \eqref{equ:proof}. Here $C,c$ are universal constants.

For the second probability, we have
\begin{eqnarray*}
\notag\mathbb P_{\mathcal D_v}(E_1^c)&\leq&\mathbb P_{\mathcal D_v}(\hat{\CR}(\PI_{\bar j})< 1-\alpha_{\min} + \frac{ q_{1-\beta}\hat\sigma_{\bar{j}}}{\sqrt{n_v}})\text{\ \ where $\bar{j}$ is the index such that }\CR(\PI_{\bar j})=1-\underline{\alpha}\\
&\leq &\mathbb P_{\mathcal D_v}(\hat{\CR}(\PI_{\bar j})< 1-\alpha_{\min} + \frac{C_1\sqrt{\log(m/\beta)}}{\sqrt{n_v}})\\
&&\text{since } \hat\sigma_{\bar{j}}\le 1 \ \text{and}\  q_{1-\beta}\leq C_1\sqrt{\log(m/\beta)} \text{\ by concentration of the maximum of Gaussian}\\
&\leq&\exp\big(-\frac{n_v\epsilon^2}{2(\underline{\alpha}(1-\underline{\alpha})+\epsilon/3)}\big)\\
&&\text{by using Bennett's inequality (e.g., equation (2.10) in \cite{boucheron2013concentration}) and noticing the variance bound in \eqref{equ:proof}}\\
&\leq&\exp\big(-C_2n_v\min\{\epsilon,\frac{\epsilon^2}{\underline{\alpha}(1-\underline{\alpha})}\}\big)
\end{eqnarray*}
where $C_2$ is a universal constant.

For the third probability, we have
\begin{eqnarray*}
\mathbb P_{\mathcal D_v}(E_3^c)&\leq&\mathbb P_{\mathcal D_v}\big(\hat{\CR}(\PI_j)\geq \tilde{\alpha}\; \text{for some $j$ such that}\; \CR(\PI_j)\leq \tilde{\alpha}/2\big)\\
&\leq&\sum_{j:\CR(\PI_j)\leq \tilde{\alpha}/2}\mathbb P_{\mathcal D_v}(\hat{\CR}(\PI_j)\geq \tilde{\alpha})\\
&\leq&\sum_{j:\CR(\PI_j)\leq \tilde{\alpha}/2}\exp\big(-\frac{n_v(\tilde{\alpha}/2)^2}{2(\CR(\PI_j)(1-\CR(\PI_j))+\tilde{\alpha}/6)}\big)\;\text{by Bennett's inequality}\\
&\leq& m\exp\big(-\frac{n_v(\tilde{\alpha}/2)^2}{\tilde{\alpha}(1-\tilde{\alpha}/2)+\tilde{\alpha}/3}\big)\\
&\leq& m\exp(-C_3n_v\tilde{\alpha})
\end{eqnarray*}
where $C_3$ is a universal constant. Combining those bounds, we obtain the overall probability bound:
\begin{align} 
&\mathbb P_{\mathcal D_v}(\CR(\PI_{j^*_{1-\alpha_k}})\geq 1-\alpha_k\; \text{for all}\; k=1,\ldots,K)\nonumber\\
\ge &\mathbb P_{\mathcal D_v}(E_2)-\mathbb P_{\mathcal D_v}(E_1^c)-\mathbb P_{\mathcal D_v}(E_3^c)\nonumber\\
\ge& 1-\beta-C\Big(\Big(\frac{r\log^7(mn_v)}{n_v\tilde{\alpha}}\Big)^{\frac{1}{6}}+\frac{\sqrt{r}\log^2 (mn_v)}{\sqrt{n_v\tilde{\alpha}}}+m\exp\big(-c(n_v\frac{\tilde{\alpha}}{r})^{2/3}\big)\Big)\nonumber\\
&-\exp\big(-C_2n_v\min\{\epsilon,\frac{\epsilon^2}{\underline{\alpha}(1-\underline{\alpha})}\}\big)-m\exp(-C_3n_v\tilde{\alpha})\label{equ:proof2}
\end{align}

Finally, we can easily see that $m\exp(-C_3n_v\tilde{\alpha})$ and $m\exp(-c(n_v\frac{\tilde{\alpha}}{r})^{2/3})$ are dominated by $\big(\frac{r\log^7(mn_v)}{n_v\tilde{\alpha}}\big)^{1/6}$ when $\big(\frac{r\log^7(mn_v)}{n_v\tilde{\alpha}}\big)^{\frac{1}{6}}<C_4$ for a constant $C_4$ (depending on $C_3$ and $c$). Moreover, $\frac{\sqrt{r}\log^2 (mn_v)}{\sqrt{n_v\tilde{\alpha}}}$ can also be neglected, because $\frac{\sqrt{r}\log^2 (mn_v)}{\sqrt{n_v\tilde{\alpha}}}= \big(\frac{r\log^{5/2}(mn_v)}{n_v\tilde{\alpha}}\big)^{\frac{1}{3}} \big(\frac{r\log^7(mn_v)}{n_v\tilde{\alpha}}\big)^{\frac{1}{6}}\leq \big(\frac{r\log^7(mn_v)}{n_v\tilde{\alpha}}\big)^{\frac{1}{6}}$ when $\big(\frac{r\log^7(mn_v)}{n_v\tilde{\alpha}}\big)^{\frac{1}{6}}<1$. On the other hand, when $\big(\frac{r\log^7(mn_v)}{n_v\tilde{\alpha}}\big)^{\frac{1}{6}}\ge 1 \text{ or } C_4$, the first error term already exceeds $1$ (by enlarging the universal constant $C$ if necessary) and the error bound holds true trivially. Therefore the desired conclusion follows from \eqref{equ:proof2} and the above observations.
\end{proof}

\begin{proof}[Proof of Theorem \ref{optimality guarantee: normalized}]
We adopt the same notations from the proof of Theorem \ref{feasibility:normalized validator_simple}.
First consider the width. Using standard concentration bounds for sub-Gaussian variables, we write for every interval $\mathrm{PI}_j=[L_j,U_j]$ and every $\eta>0$
\begin{eqnarray*}
&&\mathbb P_{\mathcal D_v}\big(\lvert \frac{1}{n_v}\sum_{i=1}^{n_v}(U_j(x_i') - L_j(x_i')) - \mathbb E_{\pi_X}[U_j(X) - L_j(X)]\rvert >\eta\Vert H\Vert_{\psi_2}\big)\\
&\leq& \mathbb P_{\mathcal D_v}\big(\lvert \frac{1}{n_v}\sum_{i=1}^{n_v}U_j(x_i') - \mathbb E_{\pi_X}[U_j(X)]\rvert >\frac{\eta\Vert H\Vert_{\psi_2}}{2}\big) + \mathbb P_{\mathcal D_v}\big(\lvert \frac{1}{n_v}\sum_{i=1}^{n_v}L_j(x_i') - \mathbb E_{\pi_X}[L_j(X)]\rvert >\frac{\eta\Vert H\Vert_{\psi_2}}{2}\big)\\
&&\text{ \ \ by the union bound}\\
&\leq&2\exp\Big(-\frac{\eta^2\Vert H\Vert_{\psi_2}^2n_v}{4C^2\Vert U_j - \mathbb E_{\pi_X}[U_j] \Vert_{\psi_2}^2}\Big) + 2\exp\Big(-\frac{\eta^2\Vert H\Vert_{\psi_2}^2n_v}{4C^2\Vert L_j - \mathbb E_{\pi_X}[L_j] \Vert_{\psi_2}^2}\Big)\\
&&\text{ \ \ for some universal constant }C\\
&\leq&4\exp\Big(-\frac{\eta^2\Vert H\Vert_{\psi_2}^2n_v}{4C^2\Vert H \Vert_{\psi_2}^2}\Big) \text{ \ \ since }\lvert L_j - \mathbb E_{\pi_X}[L_j]\rvert, \lvert U_j - \mathbb E_{\pi_X}[U_j]\rvert\leq H\\
&=&4\exp\Big(-\frac{\eta^2n_v}{4C^2}\Big).
\end{eqnarray*}
Applying the union bound to all the $m$ candidate PIs, we have
\begin{equation*}
    \mathbb P_{\mathcal D_v}\big(\lvert \frac{1}{n_v}\sum_{i=1}^{n_v}(U_j(x_i') - L_j(x_i')) - \mathbb E_{\pi_X}[U_j(X) - L_j(X)]\rvert >\eta\Vert H\Vert_{\psi_2} \text{ for some }j=1,\ldots,m\big)\leq 4m\exp\Big(-\frac{\eta^2n_v}{4C^2}\Big)
\end{equation*}
or equivalently
\begin{equation}\label{width concentration}
    \mathbb P_{\mathcal D_v}\big(\lvert \frac{1}{n_v}\sum_{i=1}^{n_v}(U_j(x_i') - L_j(x_i')) - \mathbb E_{\pi_X}[U_j(X) - L_j(X)]\rvert >C\eta\Vert H\Vert_{\psi_2} \text{ for some }j=1,\ldots,m\big)\leq 4m\exp\Big(-\frac{\eta^2n_v}{4}\Big).
\end{equation}
Next we handle the empirical coverage rate
\begin{eqnarray*}
&&\mathbb P_{\mathcal D_v}\big(\hat{\mathrm{CR}}(\mathrm{PI}_j) - \mathrm{CR}(\mathrm{PI}_j) < -\eta + \frac{q_{1-\beta}\hat\sigma_j}{\sqrt{n_v}} \big)\\
&\leq &\mathbb P_{\mathcal D_v}\big(\hat{\mathrm{CR}}(\mathrm{PI}_j) - \mathrm{CR}(\mathrm{PI}_j) < -\eta + \frac{C_1\sqrt{\log(m/\beta)}}{\sqrt{n_v}} \big)\\
&&\text{since } \hat\sigma_{j}\le 1 \ \text{and}\  q_{1-\beta}\leq C_1\sqrt{\log(m/\beta)} \text{\ by concentration of the maximum of Gaussian}\\
&\leq& \exp \Big( -2\max\big\{\eta - C_1\sqrt{\frac{\log (m/\beta)}{n_v}}, 0\big\}^2n_v\Big)\text{ \ \ by Hoeffding's inequality since } W^{(j)}_i\in [0,1].
\end{eqnarray*}
Again applying the union bound we get for every $\eta>0$,
\begin{equation}\label{coverage concentration}
    \mathbb P_{\mathcal D_v}\big(\hat{\mathrm{CR}}(\mathrm{PI}_j) - \mathrm{CR}(\mathrm{PI}_j) < -\eta + \frac{q_{1-\beta}\hat\sigma_j}{\sqrt{n_v}} \text{ for some }j=1,\ldots,m \big)\leq m\exp \Big( -2\max\big\{\eta - C_1\sqrt{\frac{\log (m/\beta)}{n_v}}, 0\big\}^2n_v\Big).
\end{equation}
Note that by choosing both universal constants in \eqref{width concentration} and \eqref{coverage concentration} large enough, we can use the same universal constant $C$ in both. To relate to the width performance, we observe that when $\hat{\mathrm{CR}}(\mathrm{PI}_j) - \mathrm{CR}(\mathrm{PI}_j) \geq -\eta + \frac{q_{1-\beta}\hat\sigma_j}{\sqrt{n_v}}$ for all $j=1,\ldots,m$, we have that for all $\mathrm{PI}_j$ whose true coverage rate $\mathrm{CR}(\mathrm{PI}_j)\geq 1-\alpha_k+\eta$ the inequality $\hat{\mathrm{CR}}(\mathrm{PI}_j) \geq \mathrm{CR}(\mathrm{PI}_j) -\eta + \frac{q_{1-\beta}\hat\sigma_j}{\sqrt{n_v}} \geq 1-\alpha_k+ \frac{q_{1-\beta}\hat\sigma_j}{\sqrt{n_v}}$ holds (i.e., the constraint in Step 3 of Algorithm \ref{calibration:normalized} is satisfied), therefore by the optimality of each $\mathrm{PI}_{j^*_{1-\alpha_k}}$ it must hold that
\begin{equation*}
    \frac{1}{n_v}\sum_{i=1}^{n_v}\lvert U_{j^*_{1-\alpha_k}}(x_i')-L_{j^*_{1-\alpha_k}}(x_i') \rvert \leq \min_{j:\mathrm{CR(\mathrm{PI}_j)\geq 1-\alpha_k+\eta}}\frac{1}{n_v}\sum_{i=1}^{n_v}\lvert U_j(x_i')- L_j(x_i')\rvert\text{ \ \ for each $k=1,\ldots, K$}.
\end{equation*}
If we further have $\lvert \frac{1}{n_v}\sum_{i=1}^{n_v}(U_j(x_i') - L_j(x_i')) - \mathbb E_{\pi_X}[U_j(X) - L_j(X)]\rvert \leq C\eta\Vert H\Vert_{\psi_2}$ for all $j=1,\ldots,m$, then
\begin{eqnarray*}
\mathbb E_{\pi_X}[U_{j^*_{1-\alpha_k}}(X) - L_{j^*_{1-\alpha_k}}(X)] &\leq &\frac{1}{n_v}\sum_{i=1}^{n_v}\lvert U_{j^*_{1-\alpha_k}}(x_i')-L_{j^*_{1-\alpha_k}}(x_i') \rvert + C\eta\Vert H\Vert_{\psi_2}\\
&\leq& \min_{j:\mathrm{CR(\mathrm{PI}_j)\geq 1-\alpha_k+\eta}}\mathbb E_{\pi_X}[U_{j}(X) - L_{j}(X)] + 2C\eta\Vert H\Vert_{\psi_2}
\end{eqnarray*}
for every $k=1,\ldots,K$. Altogether we can conclude that
\begin{eqnarray*}
&&\mathbb P_{\mathcal D_v}\Big(\mathbb E_{\pi_X}[U_{j^*_{1-\alpha_k}}(X) - L_{j^*_{1-\alpha_k}}(X)]\leq \min_{j:\mathrm{CR(\mathrm{PI}_j)\geq 1-\alpha_k+\eta}}\mathbb E_{\pi_X}[U_{j}(X) - L_{j}(X)]+2C\eta\Vert H\Vert_{\psi_2}\text{ for all }k=1,\ldots,K\Big)\\
&\geq& \mathbb P_{\mathcal D_v}\Big(\lvert \frac{1}{n_v}\sum_{i=1}^{n_v}(U_j(x_i') - L_j(x_i')) - \mathbb E_{\pi_X}[U_j(X) - L_j(X)]\rvert \leq C\eta\Vert H\Vert_{\psi_2} \text{ for all } j=1,\ldots,m, \text{ and}\\
&&\hspace{2em} \hat{\mathrm{CR}}(\mathrm{PI}_j) - \mathrm{CR}(\mathrm{PI}_j) \geq -\eta + \frac{q_{1-\beta}\hat\sigma_j}{\sqrt{n_v}} \text{ for all } j=1,\ldots,m\Big)\\
&\geq& 1- 4m\exp\Big(-\frac{\eta^2n_v}{4}\Big) - m\exp \Big( -2\max\big\{\eta - C\sqrt{\frac{\log (m/\beta)}{n_v}}, 0\big\}^2n_v\Big)\text{ \ \ by \eqref{width concentration} and \eqref{coverage concentration}}\\
&\geq&1-8m\exp \Big( -\frac{1}{4}\max\big\{\eta - C\sqrt{\frac{\log (m/\beta)}{n_v}}, 0\big\}^2n_v\Big)
\end{eqnarray*}
where the last inequality holds because $\eta\geq \max\big\{\eta - C\sqrt{\frac{\log (m/\beta)}{n_v}}, 0\big\}$.
\end{proof}

\end{document}